\documentclass[a4paper, 12pt]{amsart}
\usepackage{amsmath, amsthm, amscd, amssymb, amsfonts, amsxtra, amssymb, latexsym, bm}
\usepackage{enumerate}
\usepackage{multirow}
\usepackage{verbatim}
\usepackage{graphicx,epsfig,tikz}
\usepackage{float}
\usepackage{hyperref}
\hypersetup{colorlinks = true,	allcolors  = blue}

\newcommand{\Z}{\mathbb{Z}}
\newcommand{\N}{\mathbb{N}}
\newcommand{\ff}{\mathbb{F}}
\newcommand{\Tr}{\operatorname{Tr}}

\newcommand{\CC}{\mathcal{C}}

\newcommand{\G}{\Gamma}
\newcommand{\Div}{{\rm Div}}

\newcommand{\sk}{\smallskip}
\newcommand{\msk}{\medskip}

\newtheorem{thm}{Theorem}[section]
\newtheorem{prop}[thm]{Proposition}
\newtheorem{lem}[thm]{Lemma}
\newtheorem{coro}[thm]{Corollary}

\theoremstyle{definition}
\newtheorem{rem}[thm]{Remark}
\newtheorem{exam}[thm]{Example}
\newtheorem{defi}[thm]{Definition}

\theoremstyle{remark}

\counterwithout{table}{section}

\usepackage{color}

\begin{document} \sloppy
\numberwithin{equation}{section}
\title[Spectral relation between ICCs and GP-graphs]{The spectral relation between irreducible cyclic codes and generalized Paley graphs}
\author{Ricardo A.\@ Podest\'a, Denis E.\@ Videla}
\dedicatory{\today}
\thanks{Partially supported by CONICET and SECyT-UNC}
\address{Ricardo A.\@ Podest\'a. FaMAF -- CIEM (CONICET), Universidad Nacional de C\'ordoba. 
	\newline Av.\@ Medina Allende 2144, Ciudad Universitaria, (5000), C\'ordoba, Argentina. \newline
{\it E-mail: podesta@famaf.unc.edu.ar}}
\address{Denis E.\@ Videla. FaMAF -- CIEM (CONICET), Universidad Nacional de C\'ordoba. 
\newline	Av.\@ Medina Allende 2144, Ciudad Universitaria, (5000), C\'ordoba, Argentina.
	\newline {\it E-mail: devidela@famaf.unc.edu.ar}}

\begin{abstract} 
Let $p$ be a prime and $\ff_q/\ff_r$ a finite field extension with $q=p^m$ and $r=p^s$. 
For any $k\mid q-1$, we consider $r$-ary irreducible cyclic codes (ICC) of the form 
	$\CC(k,q/r) = \{(\Tr_{q/r}(\gamma \omega^{ik})_{i=0}^{n-1})\}_{\gamma \in \ff_r}$, 
with $\omega$ a primitive element of $\ff_q$ and 
$ n= \tfrac{q-1}{k}$, and generalized Paley (GP) graphs 
	$\Gamma(k,q) = Cay(\ff_q, \{ x^k : x \in \ff_q^* \})$. 
We show that there is a simple closed formula relating the weight distribution of $\CC(k,q/r)$ with the spectrum of $\Gamma(k_r,q)$, where $k_r=\gcd(k, \frac{q-1}{r-1})$. 
Then, we give $Spec(\Gamma(k,q))$ explicitly for those graphs associated with irreducible 2-weight cyclic codes in the semiprimitive and exceptional cases.
Finally, we give the weight enumerators of irreducible cyclic codes associated with Hamming GP-graphs. 
\end{abstract}

\maketitle

\section{Introduction: preliminaries and results} \label{sec:1}
Roughly speaking, this work relates the weight distribution of irreducible cyclic codes with the spectrum of generalized Paley graphs in a simple and precise way. 
We take advantage of this to compute weight enumerators of certain irreducible cyclic codes from known spectrum of GP-graphs and conversely. 
In particular, we show that the weight distribution of irreducible cyclic codes and the spectrum of GP-graphs determine each other. 

\noindent \textit{Motivation.}
The connection between cyclic codes and graphs was first noticed more than 50 years ago. 
Several different relations can be found in the literature, in particular between 2-weight irreducible cyclic codes and certain graphs. 

First, some authors constructed 2-weight irreducible cyclic codes (ICC) from strongly regular graphs (SRG) and conversely. 
In fact, Delsarte (\cite{D}, 1972) constructed SRGs from 2-weight ICCs and established a direct relationship between the distance matrix of the code and the adjacency matrix of the graphs. 
Also, van Lint and Schrijver (\cite{vLSch}, 1982) gave another construction of SRGs 
from cosets of mulplicative subgroups of $\ff_q$. Moreover, Calderbank and Kantor (1987, \cite{CK}) showed that any projective 2-weight ICC can be obtained from an SRG. 
Later, Haemers, Peeters and van Rijckevorsel (\cite{HPvR}, 1999) constructed binary linear codes from the row-span of the adjacency matrix of a given regular graph. In particular, they showed that the code corresponding to the Paley graph is the quadratic residue code. 

From the early 2000's on, some authors constructed linear codes with good decoding properties (PD-sets) from the row-span of the incidence matrix (\cite{GK}, 2003) and from adjacency matrix of Paley graphs (\cite{KL}, 2004). This result was later extended by considering generalized Paley graphs (\cite{SL}, 2013). 
Finally, in 2013, the spectrum of cyclic codes with many arbitrary number of zeros was computed using the spectrum of Hermitian form graphs (\cite{LHFG}, \cite{ZZDX}). By using quadratic forms, Zhou et al.\@ (\cite{ZZDX}) extended the computations to other families of codes. 

We will study the weight distribution of arbitrary irreducible cyclic codes $\CC(k,q)$ by considering generalized Paley graphs $\G(k,q)$ associated to them. We will do this by giving a simple spectral relation 
between these two combinatorial objects, namely
	$$ \{\text{codewords/frequencies}\} \qquad \leftrightsquigarrow \qquad \{\text{eigenvalues/multiplicities}\}. $$ 

\subsection{The GP-graphs $\G(k,q)$}
We now introduce the family of Cayley graphs that we are interested in. 
If $G$ is an abelian group and $S$ is a subset of $G$ (the connection set) not containing $0$, the associated Cayley graph $\Gamma = Cay(G,S)$ is the directed graph (digraph) with vertex set $G$ and where two vertices $u,v$ form a directed edge from $u$ to $v$ in $\Gamma$ if and only if $v-u \in S$. If $S$ is symmetric ($S=-S$), then $Cay(G,S)$ is a simple (undirected) graph. 
Taking a finite field as the group $G$ and $k$-th powers in the field as connection set we get the GP-graphs.
\begin{defi} \label{def: GP-graphs}
Let $\ff_{q}$ denote any finite field of $q$ elements, say $q=p^m$ with $p$ prime and $m \in \N$. 
The \textit{generalized Paley graph} is the Cayley graph (\textit{GP-graph} for short) 
\begin{equation} \label{eq: Gkq}
	\G(k,q) := Cay(\ff_{q},R_{k}) \qquad \text{with } \qquad R_{k} := \{ x^{k} : x \in \ff_{q}^*\}.
\end{equation} 
That is, $\G(k,q)$ is the graph with vertex set $\ff_{q}$ and two vertices $u,v \in \ff_{q}$ are neighbors (directed edge) 
if and only if 
$v-u=x^k$ for some $x\in \ff_{q}^*$. 
\end{defi}

GP-graphs have been extensively studied in the past few years. 
For instance, Lim and Praeger studied their automorphism groups and characterized all GP-graphs which are Hamming graphs 
(\cite{LP}, 2009). 
Also, Pearce and Praeger characterized all GP-graphs which are Cartesian decomposable (\cite{PP}, 2019); some extensions to the directed case can be found in \cite{PV7}.
The clique number of $\G(k,q)$ was studied by Yip in \cite{Y1} (2021) and \cite{Y2} (2022).  
Finally, the diameter of these graphs and its relation with the Waring problem over finite fields was recently approached by us in \cite{PV6} (2021) and \cite{PV7} (2022).

The \textit{spectrum} of a graph $\G$, denoted $Spec(\G)$, is the spectrum of its adjacency matrix $A$ (i.e.,  
the set of eigenvalues of $A$ counted with multiplicities).
If $\Gamma$ has different eigenvalues $\lambda_0, \ldots, \lambda_t$ with multiplicities $m_0,\ldots,m_t$, we write 
as usual 
	$$ 
		Spec(\Gamma) = \{[\lambda_0]^{m_0}, \ldots, [\lambda_t]^{m_t}\}.
	$$
It is well-known that an $n$-regular graph $\G$ has $n$ as one of its eigenvalues, with multiplicity equal to the number of connected components of $\G$. 
That is, $\G$ is connected if and only if $n$ has multiplicity $1$. A graph $\G$ is said to be \textit{integral} if it has integral spectrum, that is if $Spec(\G) \subset \Z$.

In a recent work, we studied the spectrum of generalized Paley graphs (see \cite{PV8}), 
where we showed that the eigenvalues are exactly the Gaussian periods (see Theorem~\ref{teo: SpecGkq}). 
In that work, we also characterized the GP-graphs with integral spectrum: they are those GP-graphs $\G(k,q)$ that satisfy the arithmetic condition $k\mid \frac{q-1}{p-1}$. See also \cite{PV3}, \cite{PV8} and \cite{PV19}, for other spectral properties of these graphs.

\subsection{The irreducible cyclic codes $\CC(k,q/r)$}
We now introduce the family of cyclic codes that are relevant to us.
A \textit{linear code} of length $n$ over $\ff_r$ is a vector subspace $\CC$ of $\ff_r^n$. If $\CC$ has dimension $h$ and minimum Hamming distance $d$, we say that it is an $[n,h,d]$-code (it is customary to use $k$ for the dimension of a code, but we have reserved $k$ for the parameter of the graphs $\G(k,q)$ and the codes $\CC(k,q)$).

The weight of a codeword $c=(c_{0},\ldots,c_{n-1})$ is the number $w(c)$ of its nonzero coordinates. 
The \textit{weight distribution} of $\CC$, sometimes also called \textit{spectrum},
is the sequence $Spec(\CC) = (A_0,\ldots,A_{n})$ of frequencies 
	$$ A_{i} = \#\{c\in\mathcal{C} : w(c)=i\} $$ 
of each possible weight $i=0,\ldots,n$.
We also consider the \textit{weight enumerator} which is the integral polynomial defined by: 
\begin{equation} \label{eq: W(x)}
	W_{\CC}(x)=A_0 + A_1 x +\cdots+ A_{n}x^{n}.
\end{equation}
We will say that a code $\CC$ is an \textit{$\ell$-weight code}, if all the non-zero codewords of $\CC$ 
only take $\ell$ distinct non-zero values. 

A linear code $\CC$ is \textit{cyclic} if for every codeword $(c_{0},\ldots,c_{n-1})$ in $\CC$ the shifted codeword  
$(c_{1},\ldots,c_{n-1},c_{0})$ is also in $\CC$.  
If $\CC$ is a cyclic code of length $n$ over $\ff_{q}$, any word $(c_0,\ldots, c_{n-1})\in \CC$ has the associated polynomial
	$c_{0}+c_{1}x +c_{2}x^{2}+\cdots + c_{n_1}x^{n-1}$ in the ring $\ff_{q}[x]/(x^{n}-1)$.
In this case, the cyclic code $\CC$ corresponds to an ideal of $\ff_{q}[x]/(x^{n}-1)$ under this correspondence, say $\CC=\langle g(x) \rangle$. 
Thus, there exists a polynomial $h(x)$ such that 
\begin{equation} \label{eq: gh=x^n-1}
	g(x)h(x) = x^{n}-1.
\end{equation}	 
The polynomial $g(x)$ is the generator polynomial of $\CC$ and the reciprocal polynomial $h^\perp(x)= h(x^{-1})x^{\deg h}$ of $h(x)$ is the check polynomial of $\CC$ and is a generator polynomial of the dual code $\CC^\perp$, which is also cyclic.
The \textit{zeros} of a cyclic code are the roots of its check polynomial in its splitting field.

\subsubsection*{Irreducible cyclic codes and notations}
The code $\CC$ is called \textit{irreducible} (or \textit{minimal}) if the polynomial 
$h(x)$ in \eqref{eq: gh=x^n-1} is irreducible over $\ff_{q}$. In this case, $h(x)$ only has one zero $\alpha$ and 
	$h(x) = m_\alpha(x)$, 
the minimal polynomial of $\alpha$. We will abbreviate irreducible cyclic codes by ICCs.

We now fix some notations that will be used throughout. 
Let $p, s, m \in \N$ with $p$ prime, $s \mid m$, and consider
the prime powers 
\begin{equation} \label{eq: r,q}
	r=p^s \qquad \text{and} \qquad  q=p^m.
\end{equation}
Thus, we have the finite field extensions $ \ff_q/\,\ff_r$ and $\ff_r/\, \ff_p$, that is 
	$ \ff_p \subset \ff_r \subset \ff_q $.
Finally, choose a positive integer $k\mid q-1$ and define   
\begin{equation} \label{eq: k}
	n=\tfrac{q-1}{k}.
\end{equation}
Note that $(n,p)=1$. The integers $k$ and $n$ will play an important role in this work.

We will be concerned with the weight distribution of the following $r$-ary cyclic codes. 
\begin{defi}
For $k, q, r$ as in \eqref{eq: r,q}--\eqref{eq: k}, consider the $r$-ary cyclic code
\begin{equation} \label{eq: Ckqs}
	\mathcal{C}(k,q/r) := 
	\big \{ c_{\gamma} = \big( \Tr_{q/r}(\gamma\, \omega^{k i}) \big)_{i=0}^{n-1} : \gamma \in \ff_{q} \big \} \subset \ff_r^n
\end{equation} 
where $\omega$ is a primitive element of $\ff_{q}$ and $\Tr_{q/r}$ is the relative trace map from $\ff_{q}$ to $\ff_{r}$.
\end{defi}

\noindent 
\textit{Note}.
The code $\mathcal{C}(k,q/r)$ clearly depends on the choice of $\omega$. However, for simplicity, we remove it from the notation since the codes obtained using different primitive elements are equivalent to each other.

From the classic theory of cyclic codes, we have that these are irreducible cyclic codes over $\ff_r$ with zero $\alpha=\omega^{-k}$, of length $n$ and dimension $h$ given by
\begin{equation} \label{eq: h}
	h = ord_{n}(r)
\end{equation}	
(i.e., the least integer $t$ such that $r^t \equiv 1 \! \pmod n$). 
That is, $\mathcal{C}(k,q/r)$ is an $[n,h]$ $r$-ary code. 
In the $p$-ary case (i.e., when $s=1$), we will denote $\CC(k,q/p)$ simply by $\CC(k,q)$. 

\subsection{Outline and results}
Briefly, we relate the spectrum of generalized Paley graphs $\G(k,q)$ with the weight distribution of the irreducible cyclic codes $\CC(k,q)$, obtaining new results for codes from graphs and viceversa.
 
Namely, in Section \ref{sec: GPs} we recall the basic properties of GP-graphs and their spectrum in terms of Gaussian periods.
In Section  \ref{sec: ICCs} we study different representations of $r$-ary ICCs and set the notations that will be used throughout the paper. 
Next, in Sections \ref{sec: spectral relation} we give the spectral relation between the ICCs and integral connected GP-graphs. This spectral correspondence between codes and graphs is not a bijection, and it is the topic of Section \ref{sec: spectral correspondence}, where we also look at weight enumerators.
In Section \ref{sec: few weights} we interpret few weight ICCs in terms of GP-graphs and obtain some new spectra of strongly regular GP-graphs from 2-weight irreducible cyclic codes.
Finally, in Section \ref{sec: Hamming} we give explicit weight enumerators for ICCs which correspond to Hamming GP-graphs.

We next give a more detailed summary of the results of the paper. 
In Section \ref{sec: ICCs}, we first show that we can always consider 
$\CC(k,q/r)$ with $n$, $q$ and $k$ satisfying the conditions \eqref{eq: r,q} and \eqref{eq: k} with $q=r^{h}$, where $h$ is as in \eqref{eq: h}. This is the content of Theorem \ref{teo: Ck Ck'}.

Section \ref{sec: spectral relation} is crucial.
In fact, we show that there is a very clear spectral relation between integral connected GP-graphs and irreducible cyclic codes. More precisely, in Theorem \ref{teo: rel Ckq Gkq} (one of the main results of the paper) 
we prove that the weight distribution of $\CC(k,q/r)$ and the spectrum of the integral connected graph $\G(k_r,q)$, where 
	$k_r = \gcd(k,\tfrac{q-1}{r-1})$, 
determine each other.
Namely, we show that the weight $w(c_\gamma)$ of the codeword $c_\gamma$ of $\CC(k,q/r)$ and the eigenvalue $\lambda_\gamma$ of $\G(k_r,q)$ 
satisfy the simple expression 
	$$ \lambda_{\gamma} = \tfrac{q-1}{k_r} -\tfrac{k r}{k_r(r-1)} w(c_{\gamma}) $$ 
or equivalently 
	$w(c_{\gamma}) = \frac{r-1}{r} (n - \frac{k_r}{k} \lambda_{\gamma})$.
Moreover, the frequencies and the multiplicities also coincide, i.e., $A_{w(c_\gamma)}= m(\lambda_\gamma)$. 
In Corollary \ref{coro: p-ary} we give this spectral relation in the $p$-ary case, where $\CC(k,q)$ corresponds to $\G(k,q)$.

Then, in Section \ref{sec: spectral correspondence} we study the relation between the weight enumerators of ICCs defined over the same finite field of the same dimension but of different lengths, obtaining a reduction formula. In Proposition \ref{prop: WC WC'} we show that 
if $\CC=\CC(k,q/r)$ and $\CC'=\CC(k',q/r)$ are ICCs of the same dimension, then 
	$$ W_\CC(x^k) = W_{\CC'}(x^{k'}) .$$
Also, this turns up to be an equivalence relation for $r$-ary cyclic codes. 

Next, in Section \ref{sec: few weights}, we use Theorem \ref{teo: rel Ckq Gkq} to translate the notion of being a few weight irreducible cyclic code in terms of GP-graphs. 
A neat correspondence is only possible for $1$-weight and $2$-weight codes. In Proposition \ref{prop: 2-weights} we show that the 1-weight ICCs correspond to the complete graph and that 2-weight ICCs correspond with semiprimitive GP-graphs. 
Schmidt and White conjectured 
that all 2-weight ICCs are of three classes: subfield subcodes, semiprimitive and exceptional. In Theorem~\ref{teo: excep codes} we give the spectrum of the eleven GP-graphs associated to the eleven exceptional ICCs.

Finally, in Section \ref{sec: Hamming},
we define Hamming ICC as those spectrally related to the Hamming GP-graphs. 
In this way, using the spectral relation between ICCs and GP-graphs and the fact that the spectrum of Hamming GP-graphs is known, in Theorem~\ref{teo Hamming} we obtain the weight distributions for some new $r$-ary irreducible cyclic codes.

\section{Preliminaries on GP-graphs and their spectrum} \label{sec: GPs}
In this short preliminary section, we recall the basic properties of GP-graphs and the characterization of their spectrum in terms of Gaussian periods.

Recall from the introduction, that the \textit{generalized Paley graph} is the Cayley graph  
	$$ \G(k,q) = Cay(\ff_{q},R_{k}) \qquad \text{with } \qquad R_{k} = \{ x^{k} : x \in \ff_{q}^*\}, $$
where $\ff_{q}$ is a finite field of $q$ elements, say $q=p^m$ with $p$ prime and $m \in \N$. 
That is, $\G(k,q)$ is the graph with vertex set $\ff_{q}$ and two vertices $u,v \in \ff_{q}$ are neighbors (directed edge) 
if and only if 
$v-u=x^k$ for some $x\in \ff_{q}^*$. 
The graphs $\Gamma(k,q)$ are also denoted $GP(q,\frac{q-1}k)$ in \cite{LP} (2009). 
However, we recently realized that 20 years before, Cohen in Section 4 of \cite{Co} (1988) defined GP-graphs, using the notation $G(q,k)$ instead of our $\G(k,q)$ and requiring $k\mid \frac{q-1}2$, which results in his graphs being undirected.
If $\omega$ is a primitive element of $\ff_{q}$, then $R_{k} = \langle \omega^{k} \rangle = \langle \omega^{(k,q-1)} \rangle$. This implies that 
	$$ \G(k,q)= \G(k',q) \qquad \text{with} \qquad k'=\gcd(k,q-1)$$
and that it is a $\frac{q-1}{\gcd(k,q-1)}$-regular graph. 
Thus, we can (and we do) always assume that $k \mid q-1$ . 
We have the following equivalences for GP-graphs: 
\begin{itemize}
	\item $\G(k,q)$ is undirected if and only if $q$ is even or else $q$ is odd and $k \mid \tfrac{q-1}2$. \sk 
	
	\item $\G(k,q)$ is connected if and only if $ord_{n}(p)=m$ where $n=\tfrac{q-1}{k}$.
\end{itemize}

\subsubsection*{Gaussian periods} 
We now recall the definition of Gaussian periods and the spectral relation with GP-graphs. 

If $q=p^m$ and $r=p^s$ with $s\mid m$, the \textit{Gaussian periods} are given by
\begin{equation}\label{eq: gaussian periods}
	\eta_{i}^{(k_r,q)} = \sum_{x\in C_{i}^{(k_r,q)}} \zeta_p^{\Tr_{q/p}(x)} \in \mathbb{C}, \qquad 0 \le i \le k_r-1,
\end{equation} 
where $\zeta_p = e^{\frac{2\pi i}{p}}$, 
	$C_{i}^{(k_r,q)} = \omega^{i} \, \langle \omega^{k_r} \rangle $
is the coset in $\ff_q^*$ of the subgroup $\langle \omega^{k_r} \rangle$ of $\ff_{q}^*$ and 
\begin{equation}\label{eq: k_r} 
	k_r = \gcd(k,\tfrac{q-1}{r-1}). 
\end{equation}   
From Theorem 13 in \cite{DY}, we have the following integrality results:
\begin{equation} \label{eq: integral Gaussian periods}
	\eta_i^{(k_r,q)} \in \Z \qquad \text{and} \qquad k_r \eta_i^{(k_r,q)} +1 \equiv 0 \pmod r.
\end{equation}
Actually, $N$ and $N_1$ are used in \cite{DY} for our $k$ and $k_r$, respectively.

\subsubsection*{The spectrum of GP-graphs}
The spectra of GP-graphs are determined by Gaussian periods. Let $n=\tfrac{q-1}k$ 
and $\eta_0 = \eta_0^{(k,q)}, \ldots, \eta_{k-1} = \eta_{k-1}^{(k,q)}$ 
be the Gaussian periods as in \eqref{eq: gaussian periods}. Denote by 
\begin{equation} \label{Gaussian periods diferentes}
	\eta_{i_1},\ldots, \eta_{i_t}
\end{equation}
the different Gaussian periods not equal to $\eta$ and define the following numbers  
\begin{equation} \label{numbers} 
	\begin{aligned}
		& \mu = \#\{0 \le \ell \le k-1 : \eta_\ell =n\} \ge 0,  \\[1mm]
		& \mu_{i_j} = \#\{ 0 \le \ell \le k-1 : \eta_{\ell} = \eta_{i_j}\} \ge 1,
	\end{aligned}
\end{equation}
for $1 \le j \le t$. For simplicity, sometimes we will use the notation $\mu =\mu_{i_0}$.

We now put together the computation of the spectrum of $\G(k,q)$ and conditions for $\G(k,q)$ to be connected and for $Spec(\G(k,q))$ to be integral, previously obtained in \cite{PV8}, which we will use later.

\begin{thm}[\cite{PV8}] 
	\label{teo: SpecGkq} 
	Let $q=p^m$ be a prime power and $k \in \N$ such that $k\mid q-1$. If we put $n=\frac{p^m-1}k$ then, in the notations in \eqref{Gaussian periods diferentes} and \eqref{numbers}, we have 
	\begin{equation} \label{spec Gkq} 
		Spec(\G(k,q)) = \{ [n]^{1+\mu n}, [\eta_{i_1}]^{\mu_{i_1}n}, \ldots, [\eta_{i_t}]^{\mu_{i_t}n} \}.
	\end{equation}
	Moreover, we have:
	\begin{enumerate}[$(a)$]
		\item $\G(k,q)$ is (strongly) connected \quad $\Leftrightarrow$ \quad $\mu=0$ \quad $\Leftrightarrow$ \quad $ord_{n}(p)=m$. \msk 
		
		\item $Spec(\G(k,q))\subset \Z$ \quad $\Leftrightarrow$ \quad $k \mid \frac{q-1}{p-1}$. 
	\end{enumerate}
\end{thm}

\begin{proof}
	It follows from Theorems 2.1 and 4.1 in \cite{PV8}.
\end{proof}

The spectrum of GP-graphs $\G(k,q)$ is explicitly known for few families. The graphs $\G(1,q)$ and $\G(2,q)$ being the complete and the Paley graphs are known. In \cite{PV8} we computed the spectrum of $\G(3,q)$ and $\G(4,q)$.  In the same work we have also computed the spectrum of Hamming and semiprimitive GP-graphs.
The spectrum of the subfamily of semiprimitive GP-graphs given by $\G(q^\ell+1,q^m)$ was obtained in \cite{PV1}.
A summary of all these computations can be found in Section~4 of \cite{PV19}.
Recently, we have studied the nature of the spectrum of this graphs in \cite{PV19}, characterizing all GP-graphs with real or integral spectrum in arithmetic terms (see \cite{PV8}). 
In \cite{PV4}, we related the spectrum of Cartesian decomposable GP-graphs with certain family of ICC.

\section{Equivalent trace representations of ICCs} \label{sec: ICCs}
Delsarte's classical theorem on linear codes states that given an extension $\ff_q/\ff_r$, if $C$ is a $q$-ary linear code, then we have the relation 
	$\mathrm{Res}(C)^\perp = \Tr(C^\perp)$ 
between $r$-ary codes.
It is a direct consequence of this theorem that any irreducible cyclic code $\CC$ is a trace code as given in \eqref{eq: Ckqs}. 
That is, if $\CC$ is an ICC over $\ff_{r}$ of length $n$ and dimension $h$, then 
	$$ \CC = \CC(k,q/r) $$ 
for some extension $\ff_q/\ff_r$, with $h$ given by \eqref{eq: h} and $k$ and $n$ related by \eqref{eq: k}. 

One can get $r$-ary irreducible cyclic codes by tracing down from different extensions $\ff_{q'}$ of $\ff_r$. 
We now show that the $r$-ary codes obtained in this way, for a given length and dimension, are 
all equal to each other and, hence, that there is a minimal extension $\ff_q$ from where we can go down, which is precisely the one defined by the parameters in \eqref{eq: r,q}, \eqref{eq: k}, and \eqref{eq: h}. We call this the \textit{minimal representation} of the irreducible cyclic code.

\begin{thm} \label{teo: Ck Ck'}  
Let $p,s, m', k', q' \in \N$ with $p$ prime such that $s\mid m'$, $k' \mid p^{m'}-1$, $q'=p^{m'}$, $r=p^{s}$.
Let $q=p^{m}$, with $m=hs$ where $h=ord_{n'}(p^s)$ and $n'=\frac{q'-1}{k'}$.
Let $k$ be a positive divisor of $q-1$, such that $n=\frac{q-1}{k}$.   
If $n=n'$, then 
\begin{equation} \label{eq: irr cyc codes}
	\CC(k',q'/r) = \CC(k,q/r).
\end{equation}	
Thus, if two $r$-ary ICCs have the same length (defined using the same primitive element), then they are equal.
\end{thm} 

\begin{proof} 
Since $n=n'$, we have that $h=ord_{n}(p^{s})$ and so $r^{\frac{m'}{s}}=p^{m'}\equiv 1\pmod{n}$. 
This implies that $m=hs\mid m'$ and so $\ff_{p^{m}}$ is a subfield of $\ff_{p^{m'}}$.

Notice that the field $\ff_{p^m}$ can be seen as a subset of $\ff_{p^{m'}}$, as the fixed elements by the Frobenius automorphism 
$\varphi_{p^{m}}(a)=a^{p^{m}}$, that is 
	$$ \ff_{p^{m}} = \{y\in \ff_{p^{m'}}: y^{p^{m}} = y \}. $$
Moreover, if $\omega$ is a primitive element of $\ff_{p^{m'}}$, 
then $\alpha=\omega^{\frac{p^{m'}-1}{p^{m}-1}}$ is a primitive element of $\ff_{p^{m}}$.
	
Let $R_{k',m'}=\{x^{k'}:x\in (\ff_{p^{m'}})^*\}$. We claim that 
\begin{equation} \label{eq: RkrFq0}
	R_{k',m'}\subseteq \ff_{p^{m}}.
\end{equation}
Since $R_{k',m'} = \langle \omega^{k'}\rangle$, 
it is enough to show that $\omega^{k'} \in \ff_{p^{m}}$. 
Since $k=\frac{p^{m}-1}{n}$ and $n=\frac{p^{m'}-1}{k'}$, we have that $k'=(\frac{p^{m'}-1}{p^{m}-1})k$ and thus 
\begin{equation}\label{eq: omk alphk'}
	\omega^{k'}=\omega^{(\frac{p^{m'}-1}{p^{m}-1})k}=\alpha^{k}.
\end{equation}
Hence, $\omega^{k'} = \alpha^{k}\in \ff_{p^{m}}$ and therefore $R_{k',p^{m'}} \subseteq \ff_{p^{m}}$ as claimed. 
	
\smallskip
	
Now, let $\gamma\in \ff_{p^{m'}}$ and let us consider the word 
	$$c_\gamma=\big(\Tr_{p^{m'}/p^{s}}(\gamma \cdot \omega^{ik'})\big)_{i=0}^{n-1}\in \CC(k',q'/r),$$ 
by \eqref{eq: RkrFq0} we have that
	$$c_\gamma=\big(\Tr_{p^{m}/p^{s}}(\Tr_{p^{m'}/p^{m}}(\gamma \cdot \omega^{ik'}))\big)_{i=0}^{n-1}=
	\big(\Tr_{p^{m}/p^{s}}(\gamma' \cdot \omega^{ik'})\big)_{i=0}^{n-1}$$
with $\omega'=\Tr_{p^{m'}/p^{m}}(\gamma)$. Thus, the equation \eqref{eq: omk alphk'} implies that
	$$c_\gamma=\big(\Tr_{p^{m}/p^{s}}(\gamma' \cdot \alpha^{ik})\big)_{i=0}^{n-1}\in \CC(k,q/r),$$
and hence $\CC(k',q'/r)\subseteq \CC(k,q/r)$.
	
Conversely, consider $c_{\gamma'}=\big(\Tr_{p^{m}/p^{s}}(\gamma' \alpha^{ik})\big)_{i=0}^{n-1}\in \CC(k,q/r)$ for some 
$\gamma'\in \ff_{p^{m}}$. 
By \eqref{eq: omk alphk'}, we have that  $\alpha^{ik}=\omega^{ik'}\in \ff_{p^{m}}$ for any $i=0,\ldots,n-1$.
On the other hand, there exists $\gamma\in \ff_{p^{m'}}$ such that $\Tr_{p^{m'}/p^{m}}(\gamma)=\gamma'$, since $\Tr_{p^{m'}/p^{m}}$ is onto.
Thus we have that
	$$c_{\gamma'}=\big(\Tr_{p^{m}/p^{s}}(\Tr_{p^{m'}/p^{m}}(\gamma \omega^{ik'}))\big)_{i=0}^{n-1} = 	
	\big(\Tr_{p^{m'}/p^{s}}(\gamma \omega^{ik'})\big)_{i=0}^{n-1}\in \CC(k',p^{m'},s).$$
Therefore, we have that $\CC(k',q'/r)=\CC(k,q/r)$, as asserted. 

The final assertion follows directly from \eqref{eq: irr cyc codes}. 
Indeed, if $\CC_1$ and $\CC_2$ are two $r$-ary irreducible cyclic codes of length $n$, then by Delsarte's theorem there exist $q_1,q_2$ powers of $r$ with $q_1,q_2\equiv 1 \pmod{n}$ such that
$\CC_1=\CC(k_1,q_1/r)$ and $\CC_2=\CC(k_2,q_2/r)$ where $k_1=\frac{q_1-1}{n}$ and $k_2=\frac{q_2-1}{n}$. 
By \eqref{eq: irr cyc codes} we have that 
$$
\CC_1=\CC(k,q/r)=\CC_2
$$
where $q=r^{h}$ with $h=ord_{n}(r)$ and $k=\frac{q-1}{k}$, as asserted.
\end{proof}

\begin{rem} \label{rem: Ckr}
The above theorem implies that we can always assume that the parameters $[n,h]$ of an irreducible $r$-ary cyclic code 
satisfy the relations \eqref{eq: r,q}, \eqref{eq: k} and \eqref{eq: h}.	
\end{rem} 

\begin{exam}
Let $p=3$, $s=2$ and $n=5$, hence $r=3^2=9$. In this case $ord_{5}(3^{2})=2$. 
By Proposition \ref{teo: Ck Ck'}, for any $t\in \mathbb{N}$ we have that 
$$ \CC_9(16,3^{4}) = \CC_9(\tfrac{3^{4t}-1}{5}, 3^{4t}). $$ 
That is, the irreducible cyclic code over $\ff_9$ of length $5$ is $\CC_9(16,3^{4})$, which has dimension 2.
Although it can be obtained by tracing down from the fields $\ff_{3}^{4t}$, this is the minimal representation.
\hfill $\diamond$	
\end{exam}

\section{The spectral relation between $\CC(k,q/r)$ and $\G(k_r,q)$} \label{sec: spectral relation}
In this section, we show that there is a neat spectral relation between integral connected GP-graphs and irreducible cyclic codes. More precisely, we will show that the weight distribution of $\CC(k,q/r)$ and the spectrum of $\G(k_r,q)$, where 
$k_r = \gcd(k,\frac{q-1}{r-1})$, determine each other. The connection will be obtained through Gaussian periods.

The computation of the spectrum of (irreducible) cyclic codes is in general a difficult task. 
There are several papers on the computation of the spectra of some of these codes using exponential sums. 
Baumert and McEliece were one of the first authors to compute the spectrum in terms of Gauss sums (\cite{BMc}, \cite{BMy}, \cite{Mc}). 
Gaussian sums and Gauss periods are related and Ding, in 2009, showed (\cite{Di1}, \cite{Di2}) that the weights of the irreducible cyclic codes $\CC(k,q/r)$ can be calculated in terms of Gaussian periods. 
On the other hand, recently, we have shown in \cite{PV8} that the eigenvalues of GP-graphs are given precisely by Gaussian periods
(see Section \ref{sec: GPs}).

\subsubsection*{The spectral relation between codes and graphs} \label{subsec: 2.2}
Now, in the notations of the Introduction, we relate the weight distribution of the $r$-ary irreducible cyclic code $\CC(k,q/r)$ with the spectrum of the generalized Paley graphs $\G(k_r,q)$, taking advantage of the fact that both invariants can be expressed in terms of Gaussian periods.

We begin by an observation about the different roles played by the orders of powers of a prime $p$ modulo certain integers $n$ in ICCs and in GP-graphs. On the one hand, we have seen that $ord_n(r)$ is the dimension of the cyclic code $\CC(k,q/r)$. On the other hand, the GP-graph $\G(k,q)$ is connected if and only if $ord_n(p)=m$ where $n=\frac{q-1}{k}$ and $q=p^m$.

The following lemma establishes a relation between the orders associated to two related triples of data $(n,r,h)$ and $(n_r,p,m)$.

\begin{lem} \label{lem: ord_n conditions}
Let $p$ be a prime, $m, h, s \in \N$ such that $m=sh$ and put $q=p^m$ and $r=p^s$.
For each $k \in \N$ such that $k\mid q-1$ put $n=\tfrac{q-1}{k}$, $k_r=\gcd(k,\frac{q-1}{r-1})$
and $n_r=\frac{q-1}{k_r}$. Thus, 
$$ord_{n}(r)=h \qquad \Leftrightarrow \qquad ord_{n_r}(p)=m.$$
\end{lem}

\begin{proof}
Assume first that $ord_{n}(p^s)=h$.
Since $k_r\mid k$, we have that $n\mid n_r$ which implies that
$ord_{n_r}(p^s) \ge ord_{n}(p^s) = h$. 
On the other hand, since $n_r\mid p^{m}-1=(p^{s})^{h}-1$, then $ord_{n_r}(p^{s})\le h$, and so we obtain that
\begin{equation}\label{eq: ord n n'}
	ord_{n_r}(p^s) = ord_{n}(p^s) = h.		
\end{equation}	

Now, suppose that there exists $a\in\{1,\ldots,m-1\}$ such that $n_r\mid p^{a}-1$. 
Then, $\frac{p^m-1}{p^{a}-1}\mid k_r$ and since $k_r=\gcd(\frac{p^m-1}{p^s-1},k)$ we obtain that
$\frac{p^m-1}{p^{a}-1}\mid \frac{p^m-1}{p^s-1}$. 
Thus, we have that $p^{s}-1\mid p^{a}-1$ and so $s\mid a$.
Hence, there exists $t\in \{1,\ldots,h-1\}$ such that $a=st$ with 
$p^{st}\equiv 1 \pmod{n_r}$ which in terms of $p^s$ means that $ord_{n_r}(p^s)\le t <h$, which contradicts the equation \eqref{eq: ord n n'}.
Thus, we obtain that  $n_r\nmid p^{a}-1$ for any $a\in\{1,\ldots,m-1\}$ and therefore $ord_{n_r}(p)=m$ as asserted. 

Conversely, if $ord_{n_r}(p)=m=sh$, it is enough to see that 
$n\nmid p^{s\ell}-1$ for any $\ell \in \{1,\ldots,h-1\}$, which is equivalent to $\frac{p^{sh}-1}{p^{s\ell}-1}\nmid k$. 
So, assume that $n\mid p^{s\ell}-1$ for some $\ell \in \{1,\ldots,h-1\}$, hence 
$\frac{p^{sh}-1}{p^{s\ell}-1}\mid k$. 
Since $\frac{p^{sh}-1}{p^{s\ell}-1}\mid \frac{p^{sh}-1}{p^s-1}$, then we have
$$\tfrac{p^{sh}-1}{p^{s\ell}-1}\mid \gcd(\tfrac{p^{sh}-1}{p^s-1},k)=k_r,$$
and hence $n_r\mid p^{s\ell}-1$ for some $\ell \in \{1,\ldots,h-1\}$, 
which contradicts the fact that
$ord_{n_r}(p)=sh$. 
Therefore, $n\nmid p^{s\ell}-1$ for all $\ell \in \{1,\ldots,h-1\}$ and so we have that $ord_{n}(p^s)=h$, as asserted.
\end{proof}

The previous lemma gives a simple relation between the order conditions in irreducible cyclic codes and the connectedness of GP-graphs.

\begin{lem} \label{lem: Gkrq conexo}
If $ord_{n}(r)=h$, then the GP-graph $\G(k_r,q)$ is connected and integral.	
\end{lem}

\begin{proof}
By ($a$) in Theorem \ref{teo: SpecGkq}, $\G(k_r,p^m)$ is connected if and only if $ord_{n_r}(p)=m$. By Lemma \ref{lem: ord_n conditions} this happens if and only if $ord_n(r)=h$.

By item ($b$) in Theorem \ref{teo: SpecGkq}, $\G(k_r,p^m)$ is integral if and only if 
since $k_r \mid \frac{p^m-1}{p-1}$. But this holds, since $k_r \mid \frac{p^m-1}{p^{s}-1}$ by definition and clearly $\frac{p^m-1}{p^s-1} \mid \frac{p^m-1}{p-1}$.
\end{proof}

We are now in a position to prove  the relation between the weight distribution of $\CC(k,q/r)$ and 
the spectrum of the integral graph $\G(k_r,q)$. This is one of the main results in this work.
\begin{thm} \label{teo: rel Ckq Gkq}
Let $p$ be a prime and $q,r,k,n\in \N$ as in \eqref{eq: r,q}--\eqref{eq: k} such that $q=r^{h}$ with $h=ord_{n}(r)$.
Then, the weight distribution of the irreducible cyclic code $\CC(k,q/r)$ as in \eqref{eq: Ckqs} and the spectrum of the integral connected GP-graph $\G(k_r,q)$ as in \eqref{eq: Gkq}, where 	
	$k_r = \gcd(k,\tfrac{q-1}{r-1})$, determine each other. 
More precisely, for any $\gamma\in \ff_{q}$, the eigenvalue $\lambda_{\gamma}$ of $\G(k_r,q)$ and the weight $w(c_\gamma)$ of $c_{\gamma} \in \CC(k,q/r)$ are related by 
\begin{equation} \label{Pesoaut}
	\lambda_{\gamma} = \frac{q-1}{k_r} -\frac{k r}{k_r(r-1)} w(c_{\gamma}) 
\end{equation}
or equivalently by 
		$$ w(c_{\gamma}) = \tfrac{r-1}{r} ( n -\tfrac{k_r}{k} \lambda_{\gamma}); $$
and the multiplicity $m(\lambda_{\gamma})$ of $\lambda_{\gamma}$ is the frequency $A_{w(c_{\gamma})}$ of $w(c_\gamma)$ for all $\gamma\in \ff_q$. 
In particular, the multiplicity of $\lambda_0 = \frac{q-1}{k_r}$ is $A_0=1$. 
\end{thm} 

\begin{proof}
Let $c_{\gamma} \in \CC(k,q/r)$ for $\gamma\in \ff_{q}$.  
Thus, if $\gamma=0 $ then $w(c_{\gamma})=w(0)=0$, by \eqref{eq: Ckqs}.
On the other hand, if $\gamma\neq 0$, then $\gamma \in C_{i}^{(k_r,q)}$ (see \eqref{eq: gaussian periods}), 
and thus from equation (12) in \cite{DY} we have
\begin{equation} \label{weight c}
	w(c_{\gamma})= \tfrac{r-1}{rk} \big(q - 1 - k_r \cdot \eta_{i}^{(k_r,q)} \big).
\end{equation}
	
By Theorem \ref{teo: SpecGkq}, the eigenvalues of $\G(k_r,q)$ are 
	$$ \lambda=n_r \qquad \text{or} \qquad \lambda_{\gamma} = \eta_{i}^{(k_r,q)}$$ if $\gamma \in C_i^{(k_r,q)}$.
Putting this in \eqref{weight c} we get \eqref{Pesoaut} for $\gamma \ne 0$. 
But \eqref{Pesoaut} also holds for $\lambda=\frac{q-1}{k_r}=n_r$ and $w=0$, as we wanted.
	
The assertion about the multiplicities of the eigenvalues $\lambda_{\gamma}$ with $\gamma \ne 0$ is clear. 
For $\gamma=0$, we have $c_\gamma=0$, so $\lambda_0=n_r$. 
Every $n$-regular graph $\G$ has $n$ as one of its eigenvalues, with multiplicity equal to the number of connected components of the graph. 
Therefore, since $\G$ is connected by Lemma \ref{lem: ord_n conditions}, the multiplicity of $\lambda_0$ equals $A_{w(0)}=A_0=1$. 
\end{proof}

As a direct consequence, we obtain the following result which relates the weigths of two ICCs corresponding to the same GP-graph.

\begin{coro} \label{coro: WC WC' weights}
Let $p$ be a prime, $q=p^m, r=p^s$ with $s\mid m$ and $k,k'$ different divisors of $q-1$. 
Let $\CC=\CC(k,q/r)$ and $\CC'= \CC(k',q/r)$ be two $r$-ary ICCs as defined in \eqref{eq: Ckqs} 
of dimension $h=ord_{n}(p^{s})=ord_{n'}(p^{s})$ and lengths $n$ and $n'$, respectively. 
If $k_r=k_r'$, in the notation of \eqref{eq: k_r}, then 
\begin{equation} \label{eq: rel k k' weights}
	w(c_{\gamma}) = \tfrac{k'}{k} w(c'_{\gamma}),
\end{equation}
where $c_{\gamma}\in \CC$ and $c'_{\gamma}\in \CC$ are the words corresponding to $\gamma \in \ff_{q}$.
\end{coro}

\begin{proof}
Since $\gcd(\frac{q-1}{r-1},k)=\gcd(\frac{q-1}{r-1},k')=k_r$, by Theorem \ref{teo: rel Ckq Gkq} we have that
	$$	\lambda_{\gamma} = \tfrac{q-1}{k_r} -\tfrac{k r}{k_r(r-1)} w(c_{\gamma})= \tfrac{q-1}{k_r} -\tfrac{k' r}{k_r(r-1)} 			 w(c'_{\gamma}) $$
where $c_{\gamma}\in \CC$ and $c'_{\gamma}\in \CC$. By a straightforward computation we arrive to
	$$	\tfrac{k}{k_r}w(c_{\gamma})=\tfrac{k'}{k_r}w(c'_{\gamma}), $$
which implies \eqref{eq: rel k k' weights}, 
as asserted.
\end{proof}

As a particular case, we obtain the following result in the $p$-ary case, relating the weight distribution of $\CC(k,q)$ with the spectrum of the GP-graph $\G(k,q)$.
We state it separately since the $p$-ary case is the most common in the literature. 

\begin{coro}[$p$-ary case] \label{coro: p-ary}
Let $p, n \in \N$ with $p$ prime and $(n,p)=1$. 
If $q=p^{m}$ with $m=ord_{n}(p)$ and $k=\frac{q-1}{n}$ with $k \mid\frac{q-1}{p-1}$,
then the spectrum of $\G(k,q)$ and the weight distribution of $\CC(k,q)$ determine each other.
More precisely, if $\gamma\in \ff_{q}$, then the eigenvalue $\lambda_{\gamma}$ of 
$\G(k,q)$ and the weight of $c_{\gamma} \in \CC(k,q)$ satisfy
\begin{equation} \label{Pesoaut2}
	\lambda_{\gamma} = n -\tfrac{p}{p-1} w(c_{\gamma}) 
	\qquad \text{and} \qquad  m(\lambda_{\gamma}) = A_{w(c_\gamma)},
\end{equation}
where $m(\lambda_{\gamma})$ is the multiplicity of $\lambda_{\gamma}$ and $A_{w(c_\gamma)}$ is the frequency of $w(c_\gamma)$.
\end{coro}

We finish the section with two important remarks relative to the spectral correspondence that we have exhibited.
 
\begin{rem}[\textit{$t$-weight codes and graphs}]
An obvious consequence of the previous theorem is that, if the $\CC(k,q/r)$ is a $t$-weight code, then $\G(k_r,q)$ has $t+1$ different eigenvalues. Suppose that the eigenvalues of $\G(k_r,q)$ are
	$$\lambda_1 > \lambda_2 > \cdots > \lambda_{t+1}$$
and the weights of $\CC(k,q/r)$ are
	$$ w_0 < w_1 < \cdots < w_t.$$
We saw that $w_i$ is related with $\lambda_{i+1}$ for $i=0,\ldots,t$. For instance, the trivial weight $w_0=0$ corresponds with the principal eigenvalue $\lambda_1=n$ (the regularity degree), while the minimal weight $w_1=w_{min}=d$ (the minimum distance) corresponds with the second largest eigenvalue $\lambda_2$, and the maximal weight $w_t=w_{max}$ with the smallest eigenvalue $\lambda_{t+1}$.
\end{rem}

\begin{rem}[\textit{The relation between ICCs and GP-graphs}] \label{rem: ICCs GPGs}
Not every GP-graph $\G(k,q)$ corresponds to an irreducible cyclic code. 
Indeed, by Lemma~\ref{lem: ord_n conditions} and Theorem~\ref{teo: rel Ckq Gkq},
the GP-graph $\G(k,q)$ (with $t+1$ eigenvalues), corresponding to a ($t$-weight) irreducible code, is connected and integral.
Hence,
	$$ \big\{ \text{Irreducible cyclic codes} \big\} \quad \rightsquigarrow \quad \big\{ \text{Integral connected GP-graphs} \big\}$$
Furthermore, given $\G(k,q)$ a connected integral GP-graph,  hence with $k\mid\frac{q-1}{p-1}$, there is an irreducible cyclic code (not unique, see Theorem \ref{teo: r-ary corresp}) which corresponds to $\G(k,q)$. Indeed, the $p$-ary irreducible cyclic code $\CC(k,q)$ is one of them. To deepen into this correspondence is the goal of the next section.
\end{rem}

\section{The spectral correspondence and weight enumerators} \label{sec: spectral correspondence}
In this section, we see that the spectral correspondence between irreducible cyclic codes and GP-graphs is not a bijection at all.
Fix an extension $\ff_r$ of $\ff_p$. 
We will first characterize all those $r$-ary ICCs  
spectrally corresponding to an integral connected given GP-graph $\G(k,q)$. Then, we give this relation in terms of weight enumerators of the codes.

\subsection{The spectral correspondence} \label{subsec: correspondence}
To give the spectral correspondence in detail, we first, we need the following technical result.

\begin{lem} \label{eq: order correspondence r}
	Let $q=p^m=r^{h}$ be a prime power and let $n\mid q-1$ such that $m=ord_{n}(p)$, 
	and let $k=\frac{q-1}{n}$ with $k \mid\frac{q-1}{r-1}$.
	If $k' n'=p^m-1$ for some $k',n' \in \N$ such that $k=\gcd(\frac{q-1}{p-1},k')$,
	then $k'=kt$ with $t\mid r-1$, $\gcd(t, \tfrac{n}{r-1})=1$ and $ord_{n'}(r)=h$. 
\end{lem}

\begin{proof}
Since $\gcd(k',\frac{q-1}{r-1})=k$ we have that $k\mid k'$ and hence there exists $t\in \mathbb{N}$ such that $k'=kt$, it is enough to see that $t$ is a divisor of $r-1$ and coprime with $\frac{n}{r-1}$.
	
Since $k\mid \frac{q-1}{r-1}$, we have that $\frac{q-1}{r-1}=kT$ for some $T\in \mathbb{N}$, moreover from the definition of $n$ we obtain that $T=\frac{n}{r-1}$.
Hence, we have 
	$$
	k=\gcd(k',\tfrac{q-1}{r-1})=\gcd(kt,kT)=k\gcd(t,T)
	$$
which implies that $\gcd(t,T)=1$ and so $\gcd(t,\frac{n}{r-1})=1$, as asserted.
	
On the other hand, since $k'\mid q-1$ we have that $q-1=k'\ell=kt\ell$ for some $\ell \in \mathbb{N}$, which implies that 
	$$
	t\ell=\tfrac{q-1}{k(r-1)}(r-1)=\tfrac{n}{r-1}(r-1)
	$$
hence $t\mid \frac{n}{r-1}(r-1)$, since $\gcd(t,\frac{n}{r-1})=1$ we obtain that $t\mid r-1$ as desired.
	
Finally, since $k=\gcd(k',\frac{q-1}{r-1})$, then $n'_r=\frac{q-1}{k'_r}=\frac{q-1}{k}=n$ and so by hypothesis we obtain that
	$$
	ord_{n'_r}(p)=ord_{n}(p)=m.
	$$
Therefore, Lemma \ref{lem: ord_n conditions} implies that $ord_{n'}(r)=h$, as asserted.
\end{proof}

In order to study further the spectral correspondence given by Theorem \ref{teo: rel Ckq Gkq}, we need the following definition.
\begin{defi}
Let $q=p^{m}$ be a prime power of $p$ and let $k\in \mathbb{N}$ be a divisor of $\frac{q-1}{p-1}$. We define the \textit{primitive height of $k$ with respect to $p$}, denoted by 
	$$ h_p(k) $$ 
to be the maximum $s\in \N$ such that $k\mid \frac{q-1}{p^{s}-1}$.
\end{defi}
Notice that $h_p(k)$ must necessarily be a divisor of $m$ in the above definition.
Given $t\in \N$, we denote by ${\Div}^{+}(t)$ be the set of all positive divisors of $t$. 

We can now give the promised relation. More precisely, the next result describes all $r$-ary irreducible cyclic codes which are spectrally related with a single integral connected GP-graph.

\begin{thm} \label{teo: r-ary corresp} 
Let $\G(k,q)$ be an integral connected GP-graph with $q=p^{m}$, that is $kn=q-1$ with $q=p^{m}$, $k\mid \frac{q-1}{p-1}$ and $n\dagger q-1$.
In the spectral correspondence given by Theorem~\ref{teo: rel Ckq Gkq}, 
the graph $\G(k,q)$ determines the spectra of all the $r$-ary irreducible cyclic codes in the set 
	$$ \big\{ \CC(kt,q) \,:\, t \mid r-1 \text{ and } \gcd(t,\tfrac{n}{r-1})=1 \big\},$$ 
where $r=p^{s}$ with $s\in \Div^+(h_p(k))$.
\end{thm}

\begin{proof}
For any $s \in \Div^{+}(h_{p}(k))$, since $k\mid \frac{q-1}{p^{s}-1}$, we have that 
	$$
		\gcd(k,\tfrac{q-1}{p^{s}-1})=k.
	$$
Notice that $\gcd(k,\tfrac{q-1}{p^{s}-1})\neq k$ for any $s\not \in \Div^{+}(h_{p}(k))$. 
Thus, given $s\not \in \Div^{+}(h_{p}(k))$ any $p^{s}$-ary ICC is not spectral related with $\G(k,q)$. 
So, we can assume that $s\in \Div^{+}(h_{p}(k))$. 

Now, given $s\in \Div^{+}(h_p(k))$, notice that for all the $r$-ary irreducible cyclic codes $\CC(kt,q)$ with $t \mid r-1$,  $\gcd(t,\frac{n}{r-1})=1$ and  $r=p^{s}$, the parameters yield
	$$
		\gcd(kt, \tfrac{q-1}{r-1}) = \gcd(kt, k\tfrac{n}{r-1}) = k\gcd(t,\tfrac{n}{r-1}) = k.
	$$
Since $ord_{n'}(r)=h$ with $q=r^h$, Theorem \ref{teo: rel Ckq Gkq} implies that the spectra of $\CC(kt,q)$ is determined by the GP-graph $\G(k,q)$.

On the other hand, if $\CC(k',q/r)$ is an irreducible cyclic code of length  $n'=\frac{q-1}{k}$ and dimension $h=ord_{n'}(r)$ satisfying that $q=r^h$ such that its spectra is determined by $\G(k,q)$, then we have that
	$$
		\gcd(k',\tfrac{q-1}{r-1})=k.
	$$
Thus, Lemma \ref{eq: order correspondence r} implies that 
	$$ \CC(k',q/r) = \CC(kt,q) $$ 
with $t \mid r-1$ and $\gcd(t,\frac{n}{r-1})=1$, 
where $r=p^{s}$ with $s\in \Div_p^+(h(k))$, as asserted.
\end{proof}

We now illustrate the corollary in a simple situation with $p$-ary codes.

\begin{exam} \label{exam: C(19,343)}
Let $p=7$ and $m=3$, the set of positive divisors of $p-1=6$ is $\{1,2,3,6\}$. If we take $k=19$ then $19$ divides $\frac{7^{3}-1}{7-1}=57$ and $18=\frac{342}{19}$ is a primitive divisor of $7^{3}-1$ which implies that the GP-graph 
	$\G(19,343)$ 
is integral and connected. 
By Corollary \ref{teo: r-ary corresp}, the spectrum of $\G(19,343)$ determines the spectra of the following $7$-ary irreducible cyclic codes 
	$\CC(19,343)$ and 
	$\CC(38,343)$. 
Notice that these codes have different lengths, namely $18$ and $9$, respectively.
\hfill $\diamond$
\end{exam}

\subsection{Weight enumerators} \label{subsec: W(x)}
Now, we study the equivalence of irreducible cyclic codes, in the spectral correspondence with GP-graphs, in terms of weight enumerators.

First, as a direct consequence of the main Theorem \ref{teo: rel Ckq Gkq}, we have the following relation between the weight enumerators of two $r$-ary irreducible cyclic codes of different lengths but the same dimension. 

\goodbreak 

\begin{prop} \label{prop: WC WC'}
Let $p$ be a prime, 
$q=p^m, r=p^s$ with $s\mid m$ and $k,k'$ different divisors of $q-1$. 
Let $\CC=\CC(k,q/r)$ and $\CC'= \CC(k',q/r)$ be two $r$-ary ICCs as defined in \eqref{eq: Ckqs} 
of dimension $h=ord_{n}(p^{s})=ord_{n'}(p^{s})$ and lengths $n$ and $n'$, respectively. 
If $k_r=k_r'$, in the notation of \eqref{eq: k_r},
then 
\begin{equation} \label{eq: rel k k'}
	W_{\CC}(x^{k}) = W_{\CC'}(x^{k'}) ,
\end{equation}
or, equivalently, 
$W_{\CC}(x) = W_{\CC'}(x^{\frac{k'}k})$.
\end{prop}

\begin{proof}
By Corollary \ref{coro: WC WC' weights}, we know that
	$w(c_{\gamma})=\tfrac{k'}{k}w(c'_{\gamma})$.
Now, by Theorem \ref{teo: rel Ckq Gkq} we have that the frequencies of $w(c_{\gamma})$ and $w(c_{\gamma})$ are equal to the multiplicity of the eigenvalue $\lambda_{\gamma}$ of $\G(k_r,q)$, this implies \eqref{eq: rel k k'}, as asserted.
\end{proof}

\begin{exam} \label{exam: C19}
Let $\CC=\CC(19,343)$ and $\CC'=\CC(38,343)$ be the $7$-ary ICCs of Example~\ref{exam: C(19,343)}. 
As we saw, the spectra of these codes are determined by the GP-graph $\G(19,343)$, since in this case $\gcd(19,\frac{7^{3}-1}{7-1})=\gcd(38,\frac{7^{3}-1}{7-1})=19$. By the above proposition we have that
	$$
		W_{\CC}(x)=W_{\CC'}(x^{2}).
	$$
Thus, the weight enumerator of the $[18,3]$ cyclic code $\CC$ is obtained from the weight enumerator of the smaller $[9,3]$ cyclic code $\CC'$, just by squaring the variable.
\hfill $\diamond$
\end{exam}

\begin{rem} \label{rem: rel equiv}
Note that Proposition \ref{prop: WC WC'} defines a relation for $r$-ary linear codes. Namely, 
	$$ \CC \sim_W \CC' \qquad \Leftrightarrow \qquad W_{\CC}(x^{a})= W_{\CC'}(x^{b}) \text{\quad for some $a,b \in \N$}.$$

\noindent ($i$) 
It turns out that $\sim_W$ is an equivalence relation. 
Indeed, reflexivity and symmetry are obvious and the transitivity property is a consequence of the implication
	$$
		W_{\CC_1}(x^{a})=W_{\CC_2}(x^{b}), \quad W_{\CC_2}(x^{c})=W_{\CC_3}(x^{d}) \qquad \Rightarrow  \qquad W_{\CC_1}(x^{ac})= W_{\CC_3}(x^{bd}).
	$$

\noindent ($ii$) 
This equivalence relation is weaker than the standard equivalence relation $\sim$ of linear codes. In fact, it is clear than $\CC \sim \CC'$ implies that $\CC \sim_W \CC'$. The converse, however, is not true in general. Indeed, we have 
	$$ \CC(19,343) \sim_{W} \CC(38,343) \qquad \text{but} \qquad \CC(19,343) \nsim \CC(38,343),$$ 
since they do not have the same length (see Example \ref{exam: C19}).
\end{rem}

As a direct consequence of Theorem  \ref{teo: r-ary corresp} and Proposition \ref{prop: WC WC'}, we have the following result.

\begin{coro} \label{coro: Ckqr=Ckrtqr}
Let $p$ be a prime and $q,r,k,n\in \N$ as in \eqref{eq: r,q}--\eqref{eq: k} such that $q=r^{h}$ with $h=ord_{n}(r)$.
The weight enumerator of $\CC=\CC(k,q/r)$, as in \eqref{eq: Ckqs}, satisfies
\begin{equation} \label{eq: WCx=Wcxt}
	W_{\CC}(x)=W_{\CC_{r}}(x^{k_r/k}),
\end{equation}	
where $\CC_{r}=\CC(k_{r},q/r)$ with $k_{r}=\gcd(k,\frac{q-1}{r-1})$.
\end{coro}

\begin{proof}
By Theorem \ref{teo: r-ary corresp} we have that $k=k_r t$, with $t\mid r-1$ and $\gcd(t,\frac{n}{r-1})=1$. 
By taking into account that 
	$$ \gcd(k,\tfrac{q-1}{r-1}) = k_r = \gcd(k_r,\tfrac{q-1}{r-1}), $$ 
the equation \eqref{eq: WCx=Wcxt} follows directly from Proposition \ref{prop: WC WC'}.
\end{proof}

\section{From few weight irreducible cyclic codes to graphs} \label{sec: few weights}
In this section, we use Theorem \ref{teo: rel Ckq Gkq} to translate in terms of GP-graphs the notion of being a few weight irreducible cyclic code. 
A crystalline correspondence is only possible for $1$-weight and $2$-weight codes.
 
Given $\ell \in \mathbb{N}$, a code is called an \textit{$\ell$-weight code} if it has $\ell$ non-zero different weights.  
In \cite{SW}, Schmidt and White characterized all $2$-weight irreducible cyclic codes, by assuming the generalized Riemann hypothesis (GRH), 
they reduced the general computation to the $p$-ary case $\CC(k,q)$ with  $k\mid \frac{q-1}{p-1}$. As Vega pointed out in \cite{V}
this characterization also holds in the general case, that is $\CC(k',q/r)$ where not necessarily $k'\mid \frac{q-1}{r-1}$ with $r=p^{s}$ and $s\ge 1$.

\subsection{The GP-graphs from one and two-weight ICCs}
We begin by translating in terms of graphs which irreducible cyclic codes are $\ell$-weight codes for $\ell=1,2$.
We next show that there is a neat characterization for those GP-graphs associated to one-weight and two-weight irreducible cyclic codes. Two-weight ICCs corresponds to strongly regular GP-graphs.

We recall that a \textit{strongly regular graph} with parameters $(n,k,e,d)$ is a connected $k$-regular graph with $n$ vertices, such that if $N_v$ denotes the set of neighbors of $v$ then  for any pair of vertices $v,w$, the size $|N_v \cap N_w|$ is equal to $e$ or $d$ depending on if $v$ and $w$ are neighbors or not, respectively. In terms of eigenvalues, it is well known that any strongly regular graph has $3$ different eigenvalues and conversely any connected graph with three different eigenvalues is strongly regular.

The correspondence between one-weight and two-weight ICCs and GP-graphs is as follows.
\begin{prop} \label{prop: 2-weights}
Let $p,m',k',n,s \in \N$ with $p$ prime such that $k'\mid p^{m'}-1$, $r=p^{s}$ and put $n=\tfrac{p^{m'}-1}{k'}$.  Let $\CC(k',p^{m'}/r)$ be the $r$-ary irreducible cyclic code as in \eqref{eq: Ckqs}. 
If $m=hs$ with $h=ord_{n}(p^{s})$, $q=p^{m}$, $k=\frac{q-1}{n}$, and $k_r=\gcd(\frac{q-1}{r-1},k)$, then we have the following:
\begin{enumerate}[$(a)$]
	\item $\CC(k',p^{m'}/r)$ is a one-weight code if and only if $\G(k_r,q)$ is the complete graph $K_{q}$. \sk 
		
	\item $\CC(k',p^{m'}/r)$ is a two-weight code if and only if $\G(k_r,q)$ is a strongly regular graph.	
\end{enumerate}
\end{prop}

\begin{proof} 
Recall that in general the unique connected regular graph with two different eigenvalues is the complete graph. On the other hand,
it is a well-known fact that an undirected connected graph is strongly regular if and only if it has $3$ different eigenvalues. 
Since the graph $\Gamma(k_r,q)$ is undirected and connected by Lemma \ref{lem: ord_n conditions}, 
the remaining assertions follow from Remark \ref{rem: ICCs GPGs}.  	
\end{proof}

\begin{rem}[\textit{All strongly regular GP-graphs}] \label{rem: SW conj}
Assuming the Schmidt and White's conjecture (\cite{SW}), the graph $\G(k,p^m)$ with $k\mid \frac{p^m-1}{p-1}$ is strongly regular if and only if:
	\begin{enumerate}[$(a)$]
		\item $(k,p^m)$ is a semiprimitive pair, or \msk 
		
		\item $(k,p^m)$ belongs to one of the exceptional pairs given by Schimdt and White (see \eqref{sporadic} and the commnets around for more details).
	\end{enumerate}
\end{rem}

\begin{rem}[\textit{Rao-Pinawala's ICCs}] \label{ex: RaoPinawala} 
In \cite{RP}, the authors considered some $2$-weight irreducible cyclic codes
$\CC(k,p^2)$ such that $p$ is an odd prime, $k\mid p-1$ and $k$ is even.
Since $k$ is even and $(p-1,p+1)=2$ for $p$ odd, we have that 
	$$ k_1 = \gcd(k,\tfrac{p^2-1}{p-1})=2 \qquad \text{and} \qquad p^2\equiv 1 \pmod{4}.$$ 
By Proposition \ref{prop: 2-weights}, the SRG corresponding to $\CC(k,p^2)$ as described above is  $\G(2,p^2)$, which in this case coincides with $P(p^2)$, the classic Paley graph with $p^2$ vertices.
This observation is compatible with what Vega pointed out in \cite{V}; these codes correspond to the semiprimitive ones in the characterization of Schmidt-White \cite{SW}. 
\end{rem}

\subsection{Some GP-spectra from $2$-weight exceptional $p$-ary ICCs} \label{sec:7}
Now, we compute the spectrum of the graphs associated with exceptional 2-weight irreducible cyclic codes. 

Schmidt and White gave the following list of pairs $(k,q)$, with $k$ in ascending order, 
\begin{gather} \label{sporadic}
	\begin{aligned}
		(11,3^5), \quad (19, 5^9), \quad (35,3^{12}), \quad (37, 7^9) \quad (43, 11^7), \quad (67, 17^{33}), \\
		(107,3^{53}), \quad (133, 5^{18}), \quad (163,41^{81}), \quad (323, 3^{144}), \quad (499, 5^{249}), 
	\end{aligned}
\end{gather}
such that $\CC(k,q)$ is an exceptional 2-weight irreducible cyclic code, and they conjectured that these are all such codes (\cite{SW}). We refer to them as \textit{exceptional pairs}.
The SW-conjecture is still open up to the author's knowledge. 

For all these pairs $(k,q)$, the number $n=\frac{q-1}k$ is always a primitive divisor of $q-1$.
For instance, for $(11,3^5)$, we see that $n=\frac{3^5-1}{11}=22$ is a primitive divisor of $3^5-1=242$ since $22\nmid 3^a-1$ for 
$a=1,2,3,4$. In fact, if $n$ is not a primitive divisor of $q-1$, then the code $\CC(k,q)$ would be a subfield subcode, which is not the case. Therefore, all these graphs $\G(k,q)$ are connected.

Now, we give the spectra of $\G(k,q)$ for the exceptional pairs given above.

\begin{thm} \label{teo: excep codes}
Let $(k,q)$ be one of the eleven exceptional pairs given in \eqref{sporadic}, where $q = p^m$. 
Then, the generalized Paley graph $\Gamma(k,q)$ is a connected strongly regular graph $srg(q, n, e, d)$ with $n = \frac{q-1}{k}$, whose distinct non-principal eigenvalues $\lambda_1, \lambda_2$ and their corresponding multiplicities $m_1, m_2$ are explicitly given by:
\begin{equation} \label{eq: autovalores exceptionals}
	\begin{aligned}
	\lambda_1 & = \tfrac{\epsilon t p^\theta - 1}{k}, \qquad & m_1 &= \tfrac{1}{2} \big\{ (q-1) - \tfrac{2n + (q-1)(n-d)}{\Delta} \big\}, \\ 
	\lambda_2 & = \tfrac{\epsilon (t-k) p^\theta - 1}{k}, \qquad & m_2 &= \tfrac{1}{2} \big\{ (q-1) + \tfrac{2n + (q-1)(n-d)}{\Delta} \big\}, \\ 
	\end{aligned}
\end{equation}
where the parameters $(p, m, \theta, t, \epsilon)$ for each exceptional pair are as in the tables above 
{\small 
\begin{equation} \label{tablita}
	\renewcommand{\arraystretch}{1}
	\begin{tabular}{|c|c|c|c|c|c|}
		\hline 
		$k$ & $p$ & $m$ & $\theta$ & $t$ & $\epsilon$ \\ \hline
		$11$ & $3$  & $5$  & $2$  & $5$  & $1$ \\
		$19$ & $5$  & $9$  & $4$  & $9$  & $1$ \\
		$35$ & $3$  & $12$ & $5$  & $17$ & $1$ \\
		$37$ & $7$  & $9$  & $4$  & $9$  & $1$ \\
		$43$ & $11$ & $7$  & $3$  & $21$ & $1$ \\
		$67$ & $17$ & $33$ & $16$ & $33$ & $1$ \\
		\hline
	\end{tabular}
	\qquad \quad 
	\renewcommand{\arraystretch}{1.15}
	\begin{tabular}{|r|r|r|r|r|r|}
		\hline 
		$k$ & $p$ & $m$ & $\theta$ & $t$ & $\epsilon$ \\ \hline
		$107$& $3$ & $53$& $25$ & $53$ & $1$ \\
		$133$& $5$ & $18$&  $8$ & $33$ & $-1$\\
		$163$& $41$& $81$& $40$ & $81$ & $1$ \\
		$323$& $3$ &$144$& $70$ & $161$& $1$ \\
		$499$& $5$ &$249$& $123$& $249$& $1$ \\
		\hline
	\end{tabular}
\end{equation}}
and where $\Delta = \sqrt{(e-d)^2 + 4(n-d)}$ with $d$ and $e$ computed from $\lambda_1$ and $\lambda_2$ via
	\begin{equation} \label{ed sporadic}
		d = n - \tfrac 14 \big\{(\lambda_1 - \lambda_2)^2-(\lambda_1+\lambda_2)^2 \big\}  \qquad \text{and} \qquad e = d + \lambda_1 + \lambda_2.
	\end{equation}
\end{thm}

\begin{proof}
By hypothesis, $\CC(k,q)$ is a 2-weight irreducible cyclic code whose exact nonzero weights $w_1$ and $w_2$ are provided by Corollary 3.2 and Table 1 in \cite{SW} in terms of the structural parameters of the tables in \eqref{tablita}. 
	
Since $\G(k,q)$ is connected and its associated code $\CC(k,q)$ has exactly two weights, Theorem \ref{teo: rel Ckq Gkq} ensures that $\Gamma(k,q)$ is a strongly regular graph $srg(q,n,e,d)$ (see also \cite{SW}). The relation between the eigenvalues of $\G(k,q)$ and the weights of the code, established in Theorem \ref{teo: rel Ckq Gkq}, gives $\lambda_1 = n - \frac{p}{p-1} w_1$ and $\lambda_2 = n - \frac{p}{p-1} w_2$. 
	
Substituting the explicit formulas for 
	$$ w_1 = \tfrac{1}{k} (p-1)p^{\theta -1}(p^{m-\theta}-\epsilon t) \qquad \text{and} \qquad 
	w_2 = w_1 + \epsilon (p-1)p^{\theta-1}$$ 
into these relations, we obtain:
	\begin{align*}
		\lambda_1 &= \tfrac{p^m - 1}{k} - \tfrac{p}{p-1} \{ \tfrac{1}{k} (p-1)p^{\theta -1}(p^{m-\theta}-\epsilon t) \} = \tfrac{p^m - 1 - (p^m - \epsilon t p^\theta)}{k} = \tfrac{\epsilon t p^\theta - 1}{k}, \\[1mm]
		\lambda_2 &= \lambda_1 - \tfrac{p}{p-1} \{ \epsilon (p-1)p^{\theta-1} \} = \tfrac{\epsilon t p^\theta - 1}{k} - \epsilon p^\theta = \tfrac{\epsilon (t-k) p^\theta - 1}{k}.
	\end{align*}
	
For any connected strongly regular graph $srg(q,n,e,d)$, the two non-principal eigenvalues $\lambda^\pm$ and their multiplicities $m(\lambda^\pm)$ are universally governed by the parameters $q, n, e, d$ through the classical relations
	\begin{equation} \label{multip sporadic}
		\lambda^{\pm} = \tfrac{(e-d) \pm \Delta}{2} \qquad \text{and} \qquad 
		m(\lambda^\pm) = \tfrac{1}{2} \big\{ (q-1) \mp \tfrac{2n + (q-1)(n-d)}{\Delta} \big\},
	\end{equation}
where $\Delta = \sqrt{(e-d)^2 + 4(n-d)}$. 
Matching $\lambda_1 = \lambda^+$ and $\lambda_2 = \lambda^-$, their sum and difference satisfy $\lambda_1+\lambda_2 = e-d$ and $\lambda_1 - \lambda_2 = \Delta$. Inverting these equations uniquely determines the intersection parameters $e$ and $d$ in terms of $\lambda_1$ and $\lambda_2$ as stated in \eqref{ed sporadic}.
	
Finally, by Theorem \ref{teo: rel Ckq Gkq}, the frequencies $A_{w_i}$ of the weights $w_i$ coincide with the multiplicities of the eigenvalues $m_i$. 
\end{proof}

\begin{rem}
Some authors called \textit{cyclotomic strongly regular graphs} to the strongly regular graphs of the form $Cay(\ff_{q},D)$, 
where $D$ is a union of cosets of a multiplicative subgroup of $\ff_{q}^*$ (see \cite{Tx}). 
In \cite{Tx}, some of the above spectra in the sporadic case are found 
by using index-$2$ Gauss sums.
\end{rem}

Next, using the previous result, we explicitly compute the spectrum of 8 out of 11 exceptional pairs.
\begin{exam}
In the tables below we present the spectrum of $\G(k,q)$ for the first eight exceptional pairs.
\renewcommand{\arraystretch}{1.35}
\begin{table}[H] 
	\begin{tabular}{|c|c|c|}
		\hline
		$(k,q)$ & parameters $srg(q,n,e,d)$ & $Spec(\G(k,q)) \smallsetminus \{[n]^1\}$ \\
		\hline
		$(11,3^5)$ 		& $srg(243, 22, 1, 2)$ & $\{[4]^{132}, [-5]^{110}\}$  \\ \hline
		$(19,5^9)$ 		& $srg(1{.}953{.}125, 102{.}796, 5{.}379, 5{.}412)$ & $\{[296]^{1{.}027{.}960}, [-329]^{925{.}164}\}$ \\ \hline
		$(35,3^{12})$ & $srg(531{.}441, 15{.}184, 427, 434)$ & $\{[118]^{273{.}312}, [-125]^{258{.}128}\}$ \\ \hline
		$(37,7^9)$ 		& $srg(40{.}353{.}607, 1{.}090{.}638, 282{.}771, 29{.}510)$ & $\{[584]^{30{.}537{.}864}, [-1817]^{9{.}815{.}742}\}$ \\ \hline
		$(43,11^7)$ 	& $srg(19{.}487{.}171, 453{.}190, 10{.}509, 10{.}540)$ & $\{[650]^{9{.}970{.}180}, [-681]^{9{.}516{.}990}\}$ \\ \hline
	\end{tabular}
  \caption{Spectra of $\G(k,q)$ for the first 5 exceptional pairs.} \label{tab: 1}
\end{table}

The frequencies $A_{w_i}$ of the weights $w_i$ of $\CC(k,q)$ in Table \ref{tab: 2} below are the multiplicities $m_i$ of the corresponding eigenvalues $\lambda_i$ of $\G(k,q)$ in Table \ref{tab: 1}.

\renewcommand{\arraystretch}{1.2}
\begin{table}[H] 	
	\begin{tabular}{|c|c|c|c|c|c|}  
		\hline
		weights & $\CC(11,3^5)$ & $\CC(19,5^9)$ & $\CC(35,3^{12})$ & $\CC(37,7^9)$ & $\CC(43,11^7)$ \\ \hline
		$w_1$ 	& 12 						& 82{.}000 			& 10{.}044      	 & 934{.}332     & 411{.}400 \\
		$w_2$ 	& 18 						& 82{.}500 			& 10{.}206 				 & 936{.}390     & 412{.}610 \\ \hline
	\end{tabular}
\caption{Spectra of $\CC(k,q)$ for the first 5 exceptional pairs.}
\label{tab: 2}
\end{table}

The pairs $(67, 17^{33})$, $(107, 3^{53})$ and $(133, 5^{18})$ have intermediate complexity and are given separately in Tables 3 
and 4.

\renewcommand{\arraystretch}{1.25}
\begin{table}[H] \label{table5}
	\begin{tabular}{l}
		Pair $(67, 17^{33})$ \\
		\hline
		$q= 40{.}254{.}497{.}110{.}927{.}943{.}179{.}349{.}807{.}054{.}456{.}171{.}205{.}137$ \\
		$n= 600{.}813{.}389{.}715{.}342{.}435{.}512{.}683{.}687{.}379{.}942{.}853{.}808$ \\ 
		$e = 8{.}967{.}364{.}025{.}602{.}125{.}902{.}458{.}937{.}044{.}032{.}559{.}119$ \\ 
		$d = 8{.}967{.}364{.}025{.}602{.}125{.}903{.}185{.}223{.}489{.}938{.}034{.}768$ \\ \hline
		$\lambda_1 =  23{.}967{.}452{.}714{.}880{.}696{.}416$ \\
		$m_1 = 20{.}427{.}655{.}250{.}321{.}642{.}807{.}431{.}245{.}370{.}918{.}057{.}029{.}472$  \\ \hline
		$\lambda_2 = -24{.}693{.}739{.}160{.}786{.}172{.}065$ \\
		$m_2= 19{.}826{.}841{.}860{.}606{.}300{.}371{.}918{.}561{.}683{.}538{.}114{.}175{.}664$ \\ \hline
		$w_1 = 565{.}471{.}425{.}614{.}439{.}939{.}283{.}497{.}632{.}625{.}940{.}854{.}016$ \\
		$w_2 = 565{.}471{.}425{.}614{.}439{.}939{.}329{.}296{.}401{.}450{.}097{.}906{.}704$ \\ \hline
	\end{tabular}
		\caption{Spectra of $\G(k,q)$ and $\CC(k,q)$ for the 6th exceptional pair.} \
	\end{table}

\begin{table}[H] \label{table6}
	\begin{tabular}{l}
		Pair $(107, 3^{53})$ \\
		\hline
		$q = 19{.}383{.}245{.}667{.}680{.}019{.}896{.}796{.}723$ \\ 
		$n = 181{.}151{.}828{.}669{.}906{.}728{.}007{.}446$ \\ 
		$e = 360{.}610{.}649{.}595{.}226{.}895{.}872{.}817$ \\ 
		$d = 360{.}610{.}649{.}595{.}234{.}814{.}457{.}952$ \\ \hline
		
		$\lambda_1 =  419{.}685{.}012{.}154$ \\
		$m_1 = 9{.}782{.}198{.}748{.}174{.}963{.}312{.}402{.}084$  \\ \hline
		$\lambda_2 = -427{.}603{.}597{.}289$ \\
		$m_2= 9{.}601{.}046{.}919{.}505{.}056{.}584{.}394{.}638$ \\ \hline
		$w_1 = 120{.}767{.}885{.}779{.}658{.}028{.}663{.}528$ \\
		$w_2 = 120{.}767{.}885{.}780{.}222{.}887{.}736{.}490$ \\ \hline
	\end{tabular} \qquad \qquad 
	\begin{tabular}{l}
		Pair $(133, 5^{18})$ \\
		\hline
		$q = 3{.}814{.}697{.}265{.}625$ \\ 
		$n = 28{.}681{.}934{.}328$ \\ 
		$e = 215{.}848{.}943$ \\ 
		$d = 215{.}652{.}162$ \\ \hline
		$\lambda_1 = -96{.}922$ \\
		$m_1 = 2{.}868{.}193{.}432{.}800$  \\ \hline
		$\lambda_2 = 293{.}703$ \\
		$m_2= 946{.}503{.}832{.}824$ \\ \hline
		$w_1 = 22{.}945{.}625{.}000$ \\
		$w_2 = 22{.}945{.}312{.}500$ \\ \hline
	\end{tabular}
	\caption{Spectra of $\G(k,q)$ and $\CC(k,q)$ for the 7th and 8th exceptional pairs.}
\end{table}
For the remaining 3 pairs the computations are too big and we omit them.
\hfill $\diamond$ 
\end{exam}

\begin{rem} \label{rem: 2-weights dist}
In \cite{SW}, Schmidt and White gave an expression for the weights of the (conjecturally) all 2-weight irreducible cyclic codes. 
From Theorems \ref{teo: rel Ckq Gkq} and \ref{teo: excep codes} and Theorem~5.4 of \cite{PV8}, we provide the frequencies of these weights in the semiprimitive and exceptional cases, thus completing the computation of the weight distributions of the exact 2-weight irreducible cyclic codes. 
\end{rem}

We finish the section with comments about $t$-weight cyclic codes with $t >2$.
\begin{rem}[\textit{Some few-weight ICCs}]
To obtain some few weight irreducible cyclic codes (with more than two weights) from GP-graphs there are two possibilities at hand. 

\noindent ($a$) Use the graphs $\G(k,q)$ for $k=3,4,5$ which were studied and their spectrum computed in \cite{PV8}.
More precisely, for $p\equiv 1 \pmod{k}$ and $m=kt$, the graphs $\G(k,q)$ with $q=p^{m}$ are integral connected GP-graphs with $k+1$ different eigenvalues for $k=3,4,5$ (see Theorems 3.1--3.2 and Proposition 3.4 from \cite{PV8}). 
In this case, via these GP-graphs and Theorem \ref{teo: r-ary corresp}, we can obtain ICCs with $k$-weights for $k=3,4,5$.

\noindent ($b$) Use the Hamming GP-graphs (see next section). 
The Hamming graph $H(b,p^a)$ is integral with $b+1$ eigenvalues. 
Hence, the Hamming GP-graphs $\G(\frac{p^{ab}-1}{b(p^a-1)},p^{ab})=H(b,p^a)$
for $b=2,3,4$, that is 
	$$ \G(\tfrac{p^a+1}2,p^{2a}), \qquad \G(\tfrac{p^{2a}+p^a+1}3,p^{3a}) \qquad \text{and} \qquad \G(\tfrac{(p^a+1)(p^{2a}+1)}4,p^{4a}) $$
give rise, under the spectral correspondence, to ICCs with $3,4,5$ weights, respectively (see \eqref{eq: Spec Hamming} in the next section, we leave the details).
\end{rem}

\section{Weight enumerators of Hamming ICCs} \label{sec: Hamming}
In this final section, by using the spectral relation between ICCs and GP-graphs, we obtain the weight enumerator of all the irreducible cyclic codes related to Hamming GP-graphs.

\subsubsection*{Hamming GP-graphs and Hamming ICCs}
A \textit{Hamming graph} $H(b,q)$ is a graph with vertex set $V=K^b$ where $K$ is any set of size $q$
(typically $\ff_q$ in applications), and where two $b$-tuples form and edge if and only if they differ in exactly one coordinate. Clearly, $H(b,q)$ is a connected graph.

Now, those GP-graphs which are Hamming graphs were characterized in \cite{LP}. 
Indeed,  a connected $\G(k,p^m)$ is Hamming if and only if 
	$$ k= \tfrac{p^{m}-1}n \qquad \text{with} \qquad n = b(p^{a}-1) $$ 
for some $b\mid m$ with $b>1$ and $m=ab$. In other words, the graph $\G(k,q)$ is Hamming if and only if 
	$$ k= \tfrac{p^{ab}-1}{b(p^a-1)}$$
with $b\in \N$ and $b(p^a-1)$ is a primitive divisor of $p^m-1$ (see \cite{PV7} for more details on this condition).
Moreover, in this case
\begin{equation}\label{eq: Hamming}
	\G(\tfrac{p^{ab}-1}{b(p^{a}-1)},p^{ab})\cong H(b,p^{a}).
\end{equation}

The spectra of $H(b,p^{a})$ is well-known, since 
any Hamming graph is the Cartesian product of the same complete graph of a given size, for instance 
	$$ H(b,p^{a}) \cong \square^b K_{p^{a}}.$$ 
Thus, Hamming graphs have integral spectrum given by  
\begin{equation}\label{eq: Spec Hamming}
	Spec(H(b,p^{a})) = \big\{ [ \ell\cdot p^{a}-b ]^{\binom{b}{\ell}(p^{a}-1)^{b-\ell}} : 0 \le \ell \le b \big\}.
\end{equation}	

We now define Hamming irreducible cyclic codes as those associated with the Hamming GP-graphs, which are integral and connected.
\begin{defi}
A \textit{Hamming irreducible cyclic code} is an ICC $\CC(k,q/r)$ such that the corresponding  GP-graph $\G(k_r,q)$ is Hamming.
Namely, $q= p^m = r^{h}$ where $h=ord_{n}(r)$ with $n=\frac{q-1}{k}$,
$k_r=\gcd(k,\frac{q-1}{r-1})$, $n_r=\frac{q-1}{k_r}$, with $n_r$ satisfying
	$$
		n_r=b(p^{a}-1)\qquad \text{and} \qquad m=ab
	$$	
for some positive integers $a,b$.
\end{defi}

\subsection*{Weight enumerators of Hamming ICCs}
By using the spectral relation between GP-graphs and ICCs we now give the weight enumerators of Hamming cyclic codes, showing that they have a nice reductive property.

\begin{thm}\label{teo Hamming}
Let $p$ be a prime and let $h,m,s,q,r$ be positive integers with $m=sh$, $r=p^{s}$ and $q=r^{h}$. 
Let $n$ be a positive integer coprime with $p$ such that $n\mid p^m-1$.
Let $\CC_\mathcal{H}=\CC(k,q/r)$ be the $r$-ary irreducible cyclic code as in \eqref{eq: Ckqs} where 
$k=\frac{p^{m}-1}{n}$. 
If $\CC_\mathcal{H}$ is a Hamming irreducible cyclic code for some positive integers $a,b$ with $n_r=b(p^{a}-1)$,
then the weight enumerator of $\CC$ satisfies
\begin{equation} \label{eq: weights decomp hamming}
	W_{\CC_\mathcal{H}}(x) = \Big( 1+ (p^a-1) \cdot  x^{\frac{n(p^{s}-1)p^{a-s}}{b(p^a-1)}} \Big)^{b} = 
	\sum_{t=0}^b \tbinom{b}{t}(p^a-1)^t \cdot  x^{\frac{tn(p^{s}-1)p^{a-s}}{b(p^a-1)}}.
\end{equation} 
In particular, any non-trivial weight has frequency divisible by $p^a-1$ and the minimum and maximum weight of $\CC_\mathcal{H}$ are 
	$w_{min}=\frac{n(p^{s}-1)p^{a-s}}{b(p^a-1)}$ and $w_{max}=\frac{n(p^{s}-1)p^{a-s}}{(p^a-1)}$, respectively.
\end{thm}

\begin{proof} 
Since $\CC$ is a Hamming irreducible cyclic code, we have that $ord_{n}(r)=h$. Hence, $ord_{n_r}(p)=m$ holds by Lemma \ref{lem: ord_n conditions}, i.e.\@ we have that $n_r\dagger p^{m}-1$. Since $n_r=b (p^a-1)$ and $m=ab$, by \eqref{eq: Hamming} we have that
	$$	
	\G(k_r,q)= \G(\tfrac{p^{ab}-1}{b(p^{a}-1)},p^{ab})\cong H(b,p^{a}).
	$$
By \eqref{eq: Spec Hamming}, its spectrum is 
	$Spec (H(b,p^{a})) = \big\{ [ \ell\cdot p^{a}-b ]^{\binom{b}{\ell}(p^{a}-1)^{b-\ell}}\}_{0 \le \ell \le b}$, and 
thus, by Theorem~\ref{teo: rel Ckq Gkq}, the weights of $\CC$ are given by
	$$ \tfrac{(b-\ell)n(p^{s}-1)p^{a-s}}{b(p^a-1)} $$ 
with frequencies $\tbinom{b}{b-\ell}(p^a-1)^{b-\ell}$, for $0 \le \ell \le b$.

Therefore, the weight enumerator of $\CC$ satisfies
	\begin{align*}
		W_{\CC}(x) & =\sum_{\ell=0}^b \tbinom{b}{b-\ell}(p^a-1)^{b-\ell} \,  x^{\frac{(b-\ell)n(p^{s}-1)p^{a-s}}{b(p^a-1)}}\\ 
		& =\sum_{t=0}^b \tbinom{b}{t}(p^a-1)^t \cdot  x^{\frac{tn(p^{s}-1)p^{a-s}}{b(p^a-1)}} 
		= \Big( 1+ (p^a-1) x^{\frac{n(p^{s}-1)p^{a-s}}{b(p^a-1)}} \Big)^b.
	\end{align*}

The last assertions are clear from the expressions in \eqref{eq: weights decomp hamming}.
\end{proof}

As a direct consequence, we obtain the following particular case of interest.

\begin{coro} \label{coro Hamming a=s}
Let $p$ be a prime and $m=sh$ with $s, h \in \N$. 
Let $n \in \N$ be coprime with $p$ such that
$n\mid p^m-1$.
Let $\CC_\mathcal{H}=\CC(k,p^{m}/p^s)$ be the $p^s$-ary Hamming irreducible cyclic code as in \eqref{eq: Ckqs} where 
$k=\frac{p^{m}-1}{n}$. 
Put 
$k_r=\gcd(\frac{p^m-1}{p^s-1},k)$ and let $n_r=\frac{p^m-1}{k_r}$.
If $n_r=h(p^s-1)$, then 
	\begin{equation}\label{eq: weights decomp hamming a=s}
		W_{\CC}(x)=\big(1+ (p^s-1) \cdot x^{\frac{n}{h}}\big)^{h}  = \sum_{t=0}^h \tbinom{h}{t}(p^s-1)^t \cdot  x^{\frac{tn}{h}}.
	\end{equation} 
In particular, any non-trivial weight has frequency divisible by $p^s-1$ and the minimum and maximum weight of $\CC_\mathcal{H}$ are $w_{min}=\frac{n}{h}$ and $w_{max}=n$, respectively.
\end{coro}

\begin{proof}
It is a direct consequence of Theorem \ref{teo Hamming} by taking $a=s$. 
Since in the notation of the above theorem, we obtain that $u_s=1$ and $b=h$.
\end{proof}


\begin{thebibliography}{XXX}
\bibitem{BMc}{\sc L.D.\@ Baumert, R.J.\@ McEliece}. 
\textit{Weights of irreducible cyclic codes}. 
Information and Control \textbf{20}, (1972) 158--175. 

\bibitem{BMy}{\sc L.D.\@ Baumert, J.\@ Mykkeltveit}.
\textit{Weight distributions of some irreducible cyclic codes}. 
DSN Progr.\@ Rep.\@ \textbf{16}, (1973) 128--131.

\bibitem{CK}
{\sc R.\@ Calderbank, W.\@ Kantor}.
\textit{The geometry of two-weight codes}. 
Bull.\@ London Math.\@ Soc.\@ \textbf{18}, (1986) 97--122.

\bibitem{Co}
\textsc{S.D.\@ Cohen}.
\textit{Clique numbers of Paley graphs}.
Quaest.\@ Math.\@ \textbf{11:2}, (1988) 225-231 .

\bibitem{CDS} 
\textsc {D.\@ Cvetkovic, M.\@ Doobs and H.\@ Sachs}. 
\textit{Spectra of graphs}.
Pure and Applied Mathematics, Academic Press, 1980.

\bibitem{D}{\sc P.\@ Delsarte}.
\textit{Weights of linear codes and strongly regular normed spaces}. 
Discrete Math.\@ \textbf{3}, (1972) 47--64.


\bibitem{Di1}{\sc C.\@ Ding}. 
\textit{The weight distribution of some irreducible cyclic codes}. 
IEEE Trans.\@ Inform.\@ Theory \textbf{55:3} (2009), 955--960.

\bibitem{Di2}{\sc C.\@ Ding}. 
\textit{A class of three-weight and four-weight codes}. 
C.\@ Xing, et al. (Eds.), Proc.\@ of the Second International Workshop on Coding Theory and
Cryptography. Lecture Notes in Computer Science, vol.\@ \textbf{5557}, Springer Verlag, (2009) 34--42.

\bibitem{DY}{\sc C.\@ Ding, J.\@ Yang}. 
\textit{Hamming weights in irreducible cyclic codes}.
Discrete Math.\@ \textbf{313:4}, (2013) 434--446.

\bibitem{Tx} \textsc{T.\@ Feng, Q.\@ Xiang}.
\textit{Strongly regular graphs from unions of cyclotomic classes}. 
J.\@ Combin.\@ Theory Ser.\@ B \textbf{102}, (2012) 982--995.

\bibitem{GK}{\sc D.\@ Ghinelli, J.D.\@ Key}. 
\textit{Codes from incidence matrices and line graphs of Paley graphs}.
Adv.\@ Math.\@ Comm.\@ \textbf{5}, (2011) 93--108.

\bibitem{HPvR}{\sc  W.\@ Haemers, R.\@ Peeters, J.\@ van Rijckevorsel}. 
\textit{Binary codes of strongly regular graphs}.
Design Code. Cryptogr.\@ \textbf{17}, (1999) 187--209.


\bibitem{KL}{\sc J.D.\@ Key, J.\@ Limbupasiriporn}. 
\textit{Partial permutation decoding for codes from Paley graphs}.
Cong.\@ Numer.\@ \textbf{170}, (2004) 143--155.

\bibitem{LHFG}{\sc S.\@ Li, S.\@ Hu, T.\@ Feng, G.\@ Ge}. 
\textit{The weight distribution of a class of cyclic codes related to Hermitian forms graphs}. 
IEEE Trans.\@ Inform.\@ Theory \textbf{59:5}, (2013) 3064--3067.

\bibitem{LP}{\sc T.K.\@ Lim, C.\@ Praeger}. 
\textit{On Generalised Paley Graphs and their automorphism groups}.
Michigan Math.\@ J.\@ \textbf{58}, (2009) 294--308.

\bibitem{Mc}{\sc R.J.\@ McEliece}. 
\textit{Irreducible cyclic codes and Gauss sums}.
Combinatorics in: Proc. NATO Advanced Study Inst., Breukelen, 1974. Math. Centre Tracts 55, Math. Centrum, Amsterdam, 1974, 179--196.

\bibitem{PP}{\sc G.\@ Pearce, C.\@ Praeger}. 
\textit{Generalised Paley graphs with a product structure}. 
Ann.\@ Comb.\@ \textbf{23}, (2019) 171--182.
	
\bibitem{PV2}\textsc{R.A.\@ Podest\'a, D.E.\@ Videla}.
\textit{Weight distribution of cyclic codes defined by quadratic forms and related curves}.
Rev.\@ Uni\'on Mat.\@ Argent.\@ \textbf{62:1}, (2021) 219--242.

\bibitem{PV6} \textsc{R.A.\@ Podest\'a, D.E.\@ Videla}.
\textit{The Waring's problem over finite fields through generalized Paley graphs}.
Discrete Math.\@ \textbf{344}, (2021) 112324.

\bibitem{PV7} \textsc{R.A.\@ Podest\'a, D.E.\@ Videla}.
\textit{A reduction formula for Waring numbers through generalized Paley graphs}.
J.\@ Algebr.\@ Comb.\@ \textbf{56:4}, (2022) 1255--1285.

\bibitem{PV4} \textsc{R.A.\@ Podest\'a, D.E.\@ Videla}.
\textit{The weight distribution of irreducible cyclic codes associated with decomposable generalized Paley graphs}.
Adv.\@ Math.\@ Comm.\@  \textbf{17:2}, (2023) 446--464. 

\bibitem{PV3}\textsc{R.A.\@ Podest\'a, D.E.\@ Videla}.
\textit{Generalized Paley graphs equienergetic with their complements}.
Linear Multilinear Algebra \textbf{72:3}, (2024) 488--515.  

\bibitem{PV1}\textsc{R.A.\@ Podest\'a, D.E.\@ Videla}.
\textit{Spectral properties of generalized Paley graphs of $(q^\ell+1)$-th powers and applications}.
Discrete Math.\@ Algorithms Appl.\@ \textbf{17:4}, (2025) 2450056 

\bibitem{PV8} \textsc{R.A.\@ Podest\'a, D.E.\@ Videla}.
\textit{Spectral properties of generalized Paley graphs}. 
Australas.\@ J.\@ Comb.\@ \textbf{91:3}, (2025) 326--365. 

\bibitem{PV19} \textsc{R.A.\@ Podest\'a, D.E.\@ Videla}.
\textit{The nature of the spectrum of generalized Paley graphs}. 
(2026), \href{https://doi.org/10.48550/arXiv.2604.06513}{arXiv.2604.06513}.

\bibitem{RP}\textsc{A.\@ Rao, N.\@ Pinnawala}. 
\textit{A family of two-weight irreducible cyclic codes}.
IEEE Trans.\@ Inform.\@ Theory \textbf{56:6}, (2010) 2568--2570.

\bibitem{SL}{\sc P.\@ Seneviratne, J.\@ Limbupasiriporn}. 
\textit{Permutation decoding from generalized Paley graphs}. 
Appl.\@ Algebra in Eng.\@ Comm.\@ and Computing \textbf{24}, (2013) 225--236.

\bibitem{SW}
\textsc{B.\@ Schmidt, C.\@ White}.
\textit{All two weight irreducible cyclic codes?}
Finite Fields Appl.\@ \textbf{8}, (2002) 1--17.

\bibitem{vLSch}
\textsc{J.\@H.\@ van Lint, A.\@ Schrijver}.
\textit{Construction of strongly regular graphs, two-weight codes and partial geometries by finite fields}. 
Combinatorica \textbf{1:1}, (1981) 63--73. 

\bibitem{V}{\sc G.\@ Vega}. 
\textit{A critical review and some remarks about one- and two-weight irreducible cyclic codes}.
Finite Fields Appl.\@ \textbf{33}, (2015) 1--13.


\bibitem{Y1} \textsc{C.H.\@ Yip}.
\textit{On the directions determined by Cartesian products and the clique number of generalized Paley graphs}.
Integers.\@ \textbf{21}, (2021) Paper A51 .	

\bibitem{Y2} \textsc{C.H.\@ Yip}.
\textit{On the clique number of Paley graphs of prime power order}.
Finite Fields App.\@ \textbf{77}, (2022) 101930.

\bibitem{ZZDX} {\sc  Z.\@ Zhou, A.\@ Zhang, C.\@ Ding, M.\@ Xiong}. 
\textit{The weight enumerator of three families of cyclic codes}. 
IEEE Trans.\@ Inform.\@ Theory \textbf{59:9}, (2013) 6002--6009.
\end{thebibliography}
\end{document}